\documentclass[aps,prl,reprint, twocolumn,
groupedaddress,longbibliography]{revtex4-2}

\usepackage[svgnames]{xcolor}
\usepackage{dcolumn}
\usepackage{url}
\expandafter\let\csname equation*\endcsname\relax
\expandafter\let\csname endequation*\endcsname\relax
\usepackage{amsmath,amssymb}
\usepackage{amsthm}
\usepackage{graphicx,bm,color}
\usepackage{mathrsfs}
\usepackage{dsfont}
\usepackage{mathtools}
\usepackage{bbm,bm}
\usepackage{mathdots}
\usepackage{xcolor}
\usepackage{tabularx}
\usepackage{booktabs}
\usepackage[colorlinks=true, citecolor=blue, linkcolor=blue, urlcolor=blue]{hyperref}
\usepackage{array}
\usepackage{amssymb}
\usepackage{amsfonts}
\usepackage{changes}
\usepackage{subfigure}
\usepackage{float}
\usepackage{cleveref}
\usepackage{enumitem}
\usepackage{braket}

\usepackage{mhchem}

\usepackage{changes}

\newtheorem{proposition}{Proposition}
\crefname{theorem}{Theorem}{Theorems}
\crefname{proposition}{Proposition}{Propositions}
\crefname{definition}{Definition}{Definitions}
\crefname{lemma}{Lemma}{Lemmas}
\crefname{figure}{Figure}{Figures}
\crefname{corollary}{Corollary}{Corollary}
\crefname{conjecture}{Conjecture}{Conjectures}
\crefname{section}{Section}{Sections}
\crefname{appendix}{Appendix}{Appendixes}
\crefname{observation}{Observation}{Observation}
\crefname{remark}{Remark}{Remark}
\crefname{example}{Example}{Examples}
\crefname{equation}{Eq.}{Eqs.}
\crefname{table}{Table}{Tables}
\crefname{theorem}{Theorem}{Theorems}

\begin{document}

\title{Non-symmetric GHZ states: weighted hypergraph and controlled-unitary graph representations}

 \author{Hrachya Zakaryan, Konstantinos-Rafail Revis, and Zahra Raissi}

 \affiliation{Department of Computer Science, Paderborn University, Warburger Str. 100, 33098, Paderborn, Germany}
 \affiliation{Institute for Photonic Quantum Systems (PhoQS), Paderborn University, Warburger Str. 100, 33098 Paderborn, Germany}

\begin{abstract} 
Non-symmetric GHZ states ($n$-GHZ$_\alpha$), defined by unequal superpositions of $\ket{00…0}$ and $\ket{11…1}$, naturally emerge in experiments due to decoherence, control errors, and state preparation imperfections. Despite their relevance in quantum communication, relativistic quantum information, and quantum teleportation, these states lack a stabilizer formalism and a graph representation, hindering their theoretical and experimental analysis. We establish a graph-theoretic framework for non-symmetric GHZ states, proving their local unitary (LU) equivalence to two structures: fully connected weighted hypergraphs with controlled-phase interactions and star-shaped controlled-unitary (CU) graphs. While weighted hypergraphs generally lack stabilizer descriptions, we demonstrate that non-symmetric GHZ states can be efficiently stabilized using local operations and a single ancilla, independent of system size. We extend this framework to qudit systems, constructing LU-equivalent weighted qudit hypergraphs and showing that general non-symmetric qudit GHZ states can be described as star-shaped CU graphs. Our results provide a systematic approach to characterizing and stabilizing non-symmetric multipartite entanglement in both qubit and qudit systems, with implications for quantum error correction and networked quantum protocols.
\end{abstract}  

\maketitle

The Greenberger–Horne–Zeilinger (GHZ) state and its generalization to $n$-qubit and $n$-qudit systems are fundamental resources in multipartite entanglement, with applications in quantum metrology \cite{ghzsensing1,ghzsensing2}, quantum secret sharing \cite{ghzsecret}, and quantum communication \cite{ghzcom1,ghzcom2,ghzcom3}. GHZ states belong to the well-studied family of graph states \cite{Hein1,Hein2}, where entanglement is encoded via controlled-Z ($CZ$) operations along the edges of a mathematical graph. The GHZ state is LU-equivalent to two standard graph representations: the {\it fully connected graph} and the {\it star-shaped graph} \cite{Hein2}. Furthermore, graph states can be described in the stabilizer formalism  \cite{stabilizers}, making them central to quantum error correction \cite{grapherror1,grapherror2}, quantum secret sharing \cite{secret1,secret2,secret3} and measurement-based quantum computation \cite{measurement_graph,measurement_hyper}.

A natural extension of graph states are hypergraph states, where entanglement is generated via multi-qubit controlled operations \cite{Hypergraph,Hypergraph2,Hypergraph3}. These states are used in quantum algorithms \cite{HypergraphAlgo2} and as quantum resources \cite{HypergraphAlgo1}. Weighted hypergraphs further extend this framework by assigning real-valued weights to hyperedges, allowing controlled-phase ($CZ^\alpha$) interactions. Such weighted structures arise in spin chains, lattice models, and purification protocols \cite{weighted1,weighted2,weighted3}. However, despite their significance, weighted hypergraphs generally lack a stabilizer description, limiting their applicability in quantum error correction and fault-tolerant quantum computing.

Beyond standard GHZ states, an important generalization is the class of non-symmetric GHZ states, defined as
\begin{align}
    \ket{n\text{-GHZ}_\alpha}=\cos{\frac{\alpha \pi}{2}}\ket{00\dots 0}+\sin{\frac{\alpha \pi}{2}}\ket{11\dots 1}.\nonumber
\end{align}
These states naturally arise in experiments due to decoherence, control errors, and imperfections in state preparation \cite{expghz2,expghz3,expghz4,expghz5}. Moreover, non-symmetric GHZ states have been experimentally realized in superconducting circuits and as entangled coherent states in optical systems \cite{Goss2024-superconducting-qutrit,tm,expghz, expghz5, expghz6}.
They also play key roles in relativistic quantum information, including the Unruh effect \cite{unruh}, and serve as probabilistic quantum channels in quantum teleportation \cite{ghzchannel1,ghzchannel2,ghzchannel3,ghzchannel4,ghzchannel5,ghzchannel6}. Despite their relevance, no general stabilizer formalism or graph-based representation exists for these states, restricting their theoretical characterization and experimental verification.

In this Letter, we introduce a graph-based framework for non-symmetric GHZ states by establishing their LU equivalence to {\it fully connected weighted hypergraphs}. While weighted hypergraphs are generally not stabilizer states, we show that non-symmetric GHZ states can be stabilized using only a {\it single ancilla} qubit, independent of system size. This approach provides a systematic way to characterize and stabilize non-symmetric multipartite GHZ states.

We further generalize this framework to {\it qudits}, where non-symmetric GHZ states exhibit additional degrees of freedom. We construct an LU-equivalent weighted qudit hypergraph, and demonstrate that a single ancilla suffices for its stabilization. Beyond hypergraph states, we extend the graph-state formalism by incorporating controlled-unitary ($CU$) gates, showing that arbitrary non-symmetric qudit GHZ states can be represented as {\it star-shaped $CU$ graphs}. This controlled-unitary formalism provides a novel perspective on non-symmetric  multipartite entanglement and expands graph-based quantum state analysis.

\textit{Hypergraph States.}--- Hypergraph states extend the concept of graph states by allowing multi-qubit interactions through multi-controlled-Z ($CZ$) operations. A $k$-level hypergraph $G(V,E)$ consists of a set of vertices $V$ connected by hyperedges $E$, where each hyperedge may involve up to $k$ vertices, i.e., $E\subseteq\bigoplus_{k^\prime=2}^k V^{k^\prime}$, introducing a richer entanglement structure compared to standard graph states.

A hypergraph state associated with a given hypergraph $G(V,E)$ is defined as:
 \begin{align}
    \ket{G} = \prod_{\textbf{e} \in E} CZ_{\textbf{e}}^{\Gamma_{\textbf{e}}} \ket{+}^{\otimes n},
 \end{align}
where $\ket{+}=H\ket{0}$ and each hyperedge $\textbf{e}$ contributes a multi-controlled phase gate parameterized by the adjacency matrix element $\Gamma_{\textbf{e}}$,
 \begin{align}
    \Gamma_{\textbf{e}}=\begin{cases}
            r & \text{if } (\textbf{e}=\{e_1, e_2,...,e_k\},r)\in E\\
            0 & \text{otherwise}
    \end{cases},
 \end{align}
 where for hyperedges involving fewer than $k$ vertices, we fill the corresponding index vector with zeros to maintain a uniform representation.
The multi-controlled-Z ($CZ_\textbf{e}$) operation on qudits in the set $\textbf{e}$ is given by,
 \begin{align}
    CZ_\textbf{e}=\displaystyle\sum_{k_{e_1},\dots,k_{e_m}=0}^{d-1}\omega^{\bar{k}_\textbf{e}}\ket{k_{e_1} \dots k_{e_m}}\bra{k_{e_1} \dots k_{e_m}},
 \end{align}
where $\omega=e^{2\pi i/d}$, $m=|\textbf{e}|$ is the size of the set $\textbf{e}$ and $\bar{k}_\textbf{e}$ represents the collective phase contribution of the qudit states associated with the hyperedge interaction i.e., $\bar{k}_\textbf{e}=\displaystyle\prod_{j=1}^m k_{e_j}$.

The special case of $k=2$ corresponds to graph states, where each edge represents a standard controlled-Z operation, whereas for $k>2$, hypergraph states capture a broader class of multipartite entanglement structures. This extended framework not only provides a richer mathematical foundation for quantum networks \cite{tian2024-routinghypergraph} but also offers a powerful tool for analyzing non-symmetric GHZ states, revealing new stabilizer properties and circuit representations that enhance their applicability in quantum computing \cite{Banerjee2020-blockchain-weighted-hypergraph}.
\begin{figure}
    \centering
    \includegraphics[width=0.7\linewidth]{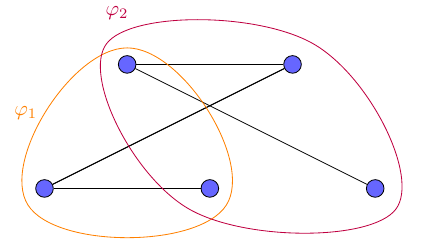}
    \caption{Representation of a weighted hypergraph state. A hyperedge of weight $\varphi_1$ connects three vertices on the left, while a separate hyperedge of weight $\varphi_2$ links four vertices on the right.}
    \label{fig:2}
\end{figure}

\textit{Non-symmetric GHZ and Weighted Hypergraphs.}---The symmetric GHZ state is local unitary (LU) equivalent to both a star-shaped graph state, where a central qubit is entangled with all others, and a fully connected graph state, where each qubit is directly connected to every other qubit. This equivalence provides an alternative graph-based description of GHZ states, widely used in entanglement studies (see Supplementary Material for details and visual representation). 

We extend this graph-based approach to non-symmetric GHZ states,  
\begin{align}
    \ket{n\text{-GHZ}_\alpha}=\cos{\frac{\alpha \pi}{2}}\ket{00\dots0}+\sin{\frac{\alpha \pi}{2}}\ket{11\dots1},
\end{align}  
which introduces an imbalance in superposition amplitudes controlled by the parameter $\alpha$. Non-symmetric GHZ states are LU equivalent to fully connected weighted hypergraph states. We formalize this equivalence in the following proposition.  

\textit{Proposition 1.}---(Weighted hypergraph representation of non-symmetric GHZ states)
For a set of vertices $V$, let $E_k$ denote the set of all $k$ element combinations from $V$.  The non-symmetric GHZ state $\ket{n\text{-GHZ}_\alpha}$ is LU equivalent to the following fully connected weighted hypergraph state:
\begin{align}\label{eq:fully-connected-hypergraph}
    \ket{G}:=\displaystyle\prod_{k=1}^{n-1}\prod_{\textbf{e}\in E_{k+1}}CZ_{\textbf{e}}^{(-2)^{k}\alpha}\ket{+}^{\otimes V},
\end{align}
where each hyperedge $\textbf{e}$ involving $k+1$ vertices carries a weight of $(-2)^{k}\alpha \pi$ (see Table~\ref{tab:1}).
 This equivalence is achieved via the local transformation:  
\begin{align}
    \ket{G}= RZ_1^\dagger(\alpha\pi)P^\dagger(\alpha\pi)^{\otimes V\setminus\{1\}}H^{\otimes V}P^\dagger_1(\pi/2)\ket{n\text{\normalfont-GHZ}_\alpha},
\end{align} 
where  
\begin{align}
    RZ(\varphi)=\begin{bmatrix}
        e^{-i\varphi/2}&0\\
        0&e^{i\varphi/2}
    \end{bmatrix}, \hspace{0.4cm}
    P(\varphi)=\begin{bmatrix}
        1 &0 \\
        0& e^{i\varphi}
    \end{bmatrix}.
\end{align}  
\begin{table}
    \centering
    \begin{tabular}{|c|}
    \hline\\
         $\ket{2\text{-GHZ}_\alpha}=\cos{\frac{\alpha \pi}{2}}\ket{00}+\sin{\frac{\alpha \pi}{2}}\ket{11}$\\[5pt]
         \includegraphics[width=0.25\linewidth]{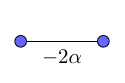}\\
        \hline\\
         
         $\ket{3\text{-GHZ}_\alpha}=\cos{\frac{\alpha \pi}{2}}\ket{000}+\sin{\frac{\alpha \pi}
         {2}}\ket{111}$\\[5pt]
          \includegraphics[width=0.7\linewidth]{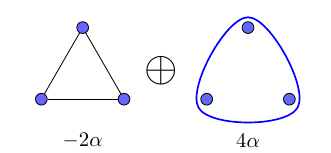}\\\hline
         \\
         
         $\ket{4\text{-GHZ}_\alpha}=\cos{\frac{\alpha \pi}{2}}\ket{0000}+\sin{\frac{\alpha \pi}{2}}\ket{1111}$\\[5pt]
           \includegraphics[width=\linewidth]{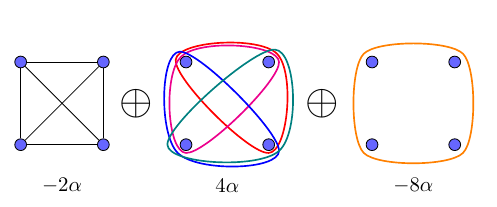}\\\hline
    \end{tabular}
    \caption{ Fully connected weighted hypergraph representations of $\ket{n\text{-GHZ}_\alpha}$ for $n = 2, 3, 4$ qubits. Each hyperedge acts on $k$ qubits and carries a phase weight of $(-2)^k \alpha \pi$.}\label{tab:1}
\end{table}

The detailed proof, provided in the Supplementary Material, uses successive local unitary transformations, including Hadamard and controlled-phase gates, to map the non-symmetric GHZ state to its hypergraph representation. For $\alpha=1/2$, hyperedges of size $k=1$ apply CZ gates, while those with $k\geq2$ reduce to identity operations.
This simplification transforms the weighted hypergraph into a fully connected graph, where each qubit is directly connected to every other qubit through $CZ$ interactions.

Although LU equivalence provides a structured framework for non-symmetric GHZ states, weighted hypergraph states lack a stabilizer description. We resolve this by constructing a complete stabilizer set with a single ancilla, independent of system size $n$, allowing direct stabilizer verification.

\textit{Local Stabilizers via a Single Ancilla.}---Unlike standard graph states, weighted hypergraph states do not admit a direct stabilizer formalism \cite{Chen_2024-hypergraph}. However, we demonstrate that by introducing a single ancilla qubit, the fully connected weighted hypergraph associated with the non-symmetric GHZ state can be efficiently stabilized using local operators.

This single ancilla, entangled with all $n$ qubits through controlled-$Z$ operations, transforms the system into a structure that admits a local stabilizer set. Importantly, the number of ancillas remains fixed at one, regardless of system size, ensuring scalability.

\textit{Proposition 2.}---(Stabilization of non-symmetric GHZ states with a single ancilla)
Given the fully connected weighted hypergraph state in Eq.~\eqref{eq:fully-connected-hypergraph}, we construct a complete stabilizer set utilizing a single ancilla qubit. The stabilizers are given by: 
\begin{align}\label{eq:stabilizers}
    \hspace{-0.1cm} K_{G^\prime}^{(1)}&=X\otimes Z\otimes Z \otimes ... \otimes Z\nonumber\\
    \hspace{-0.1cm}K_{G^\prime}^{(2)}&=X^\alpha Z X^{-\alpha}\otimes P^\dagger(\alpha\pi)XP(\alpha\pi)\otimes I \otimes ... \otimes I\nonumber\\
    \hspace{-0.1cm}K_{G^\prime}^{(3)}&=X^\alpha Z X^{-\alpha}\otimes I\otimes P^\dagger(\alpha\pi)XP(\alpha\pi) \otimes ... \otimes I\\
    \hspace{-0.1cm}&\hspace{0.15cm}\vdots\nonumber\\
    \hspace{-0.1cm}K_{G^\prime}^{(n+1)}&=X^\alpha Z X^{-\alpha}\otimes I\otimes I \otimes ... \otimes P^\dagger(\alpha\pi)XP(\alpha\pi)\, ,\nonumber
\end{align}

where 
\begin{align}
    X^\alpha \coloneqq \frac{1}{2}((1-e^{i\pi\alpha})I+(1+e^{i\pi\alpha})X)\, ,
\end{align}
the first qubit is the ancilla, and $G^\prime$ is the graph with the ancilla.

The proof (see Supplementary Material) follows a known LU transformation technique \cite{Tsimakuridze_2017}, mapping a graph state to a weighted hypergraph state. The key insight is that the introduction of an ancilla transforms the weighted interactions into an effective star-shaped graph with local stabilizers. By applying $X^\alpha$ for $\alpha \in \mathbb{R}$ on the ancilla and local phase gates on the rest we obtain the weighted hypergraph in Eq.~\eqref{eq:fully-connected-hypergraph} on the non-ancilla qubits.

\textit{Generalization to Qudits.}---For qubits, the GHZ state contains only two elements in the superposition, limiting the possible ways to introduce asymmetry. However, in qudits, the GHZ state extends over $ d $ levels,  
\begin{align}
    \ket{n\text{-GHZ}^d}=\sum_{k=0}^{d-1}\ket{kk...k},
\end{align}
allowing for a richer class of non-symmetric generalizations. For simplicity, we omit normalization factors in our expressions where they are not essential.

To define a qudit non-symmetric GHZ state that naturally connects to weighted hypergraphs, we first generalize the $X^\alpha$ operator from qubits to qudits,  
\begin{align}
    X^\alpha=H^\dagger Z^\alpha H=\sum_{k,l=0}^{d-1}\omega^{l(\alpha-k)}X^k.
\end{align}
Using this, we construct the following qudit generalization:

\textit{Proposition 3.}---(Qudit weighted hypergraph representation of non-symmetric GHZ states)
    Given a qudit $ X^\alpha $, the non-symmetric qudit GHZ state is defined as  
    \begin{align}
        \ket{n\text{-GHZ\,}^d_\alpha} \coloneqq  \prod_{k=2}^{n}CX_{1k} \, X^\alpha_1\ket{0}^{\otimes n}\, ,
    \end{align}
    which can be rewritten as  
    \begin{align}
        \ket{n\text{-GHZ\,}^d_\alpha} = \sum_{k,l=0}^{d-1}\omega^{l(\alpha-k)}\ket{kk....k}.
    \end{align}  
    This state is LU equivalent to a qudit-weighted hypergraph with controlled-phase interactions, given by  
    \begin{align}  \label{eq:qudit hypergraph}
        CZ_{G}^{-\alpha}=(Z_1Z_2...Z_{n})^{-\alpha}.
    \end{align}
where
\begin{align}
    CZ_G=\displaystyle\prod_{\textbf{e}\in E} CZ_{\textbf{e}} \text{ for } G=(V,E)
\end{align}
A detailed proof of this proposition is provided in the Supplementary Material. Unlike in the qubit case, where weighted hypergraphs emerge naturally, the qudit case introduces additional complexity. This arises from the fact that,  
\begin{align} 
(Z_2Z_3...Z_{n+1})^{-\alpha} \neq Z_2^{-\alpha}Z_3^{-\alpha}...Z_{n+1}^{-\alpha},
\end{align}  
requiring correction terms to expand the brackets. This is present for qubits as well, however while for qubits the correction terms are of the form $CZ^\alpha$, for qudits more general non-local phase interactions are required. These are given by, 
\begin{align}  
CP_e(\vec{\alpha})=\sum_{k_{e_1},...,k_{e_m}=0}^{d-1}\omega^{\vec{\alpha}_{k_e}}\ket{k_{e_1}...k_{e_m}}\bra{k_{e_1}...k_{e_m}},  
\end{align} 
where $\vec{\alpha}_{k_e}$ denotes the phase factor associated with the decimal encoding  $k_e$, of the computational basis states. Therefore, $ CP_e(\vec{\alpha}) $ is a diagonal matrix of phases determined by the weight vector $ \vec{\alpha} $. Unlike in the qubit case, where the interaction weights follow a simple recursive pattern, the qudit case requires explicit computation of the interaction structure.  This necessitates an algorithmic approach, which we detail in the Supplementary Material.

Despite this structural complexity, the stabilizer formulation remains highly analogous to the qubit case. The stabilizers for the qudit non-symmetric GHZ state take the form  
\begin{align}
    \hspace{-0.23cm}K_{G^\prime}^{(1)}&=X\otimes Z\otimes Z\otimes ... \otimes Z\nonumber\\
    \hspace{-0.23cm}K_{G^\prime}^{(2)}&=X^{-\alpha} Z X^{\alpha}\otimes P^\dagger(\alpha\pi)XP(\alpha\pi)\otimes I \otimes... \otimes I\nonumber\\
    \hspace{-0.23cm}K_{G^\prime}^{(3)}&=X^{-\alpha} Z X^{\alpha}\otimes I\otimes P^\dagger(\alpha\pi)XP(\alpha\pi)\otimes ... \otimes I\\
    \hspace{-0.23cm}&\hspace{0.15cm}\vdots\nonumber\\
    \hspace{-0.23cm}K_{G^\prime}^{(n+1)}&=X^{-\alpha} Z X^{\alpha}\otimes I\otimes I\otimes ... \otimes P^\dagger(\alpha\pi)XP(\alpha\pi).\nonumber
\end{align}

As in the qubit case, only a single ancilla is required, independent of system size $ n $. This guarantees that non-symmetric qudit GHZ states can be efficiently verified and manipulated with minimal resource overhead, making them scalable for quantum information applications.  

\textit{Controlled-Unitary Graph States.}---We extend the graph state formalism by incorporating controlled-unitary ($CU$) operations as entangling gates, in contrast to conventional diagonal gates such as controlled-Z ($CZ$) for standard graph and hypergraph states or controlled-phase ($CZ^\alpha$) gates for weighted graph states. This broader framework naturally extends hypergraph-based representations of non-symmetric GHZ states in qudit systems.

Specifically, we show that the most general non-symmetric qudit GHZ state,  
\begin{align}\label{eq:CU-GHZ}
    \ket{n\text{-GHZ}_{\vec{a}}}=\sum_{j=0}^{d-1}a_j\ket{jj\dots j},
\end{align}
where $ \vec{a} = \{a_0, a_1, ..., a_{d-1}\} \in \mathbb{C}^{d} $, can be represented as a \textit{star-shaped graph} with controlled-$ U $ gates.

\textit{Proposition 4.}---(Controlled-unitary graph representation of general non-symmetric qudit GHZ states)
For any general non-symmetric qudit GHZ state, there exists a local unitary (LU) transformation that maps $ \ket{n\text{-GHZ}_{\vec{a}}} $ to a star-shaped controlled-unitary graph state, where $U$ is given by,
\begin{align}
    U=HA^\dagger Z A H^\dagger,
\end{align}
with the operator $A$ defined as the single-qudit transformation matrix:\begin{align}
    A\ket{0}=\sum_{j=0}^{d-1}a_j\ket{j}.
\end{align}

The proof follows by explicitly constructing $\ket{n\text{-GHZ}_{\vec{a}}}$ using controlled-$X$ operations and applying a sequence of local Hadamard and phase corrections to obtain the controlled-unitary graph form (see Supplementary Material). This formulation provides a structured way to describe arbitrary non-symmetric GHZ states within a graph-based framework. The key feature of this approach is that while the control qudits remain unchanged, the target qudit undergoes the transformation $ U $. Although general $CU$ operations may not commute, in the star-shaped graph configuration, commutativity is preserved since all operations act on a shared target.

As an illustrative example, consider the qubit case where $ A = X X^\alpha $, which naturally leads to the weighted hypergraph structure described in Eq.~\eqref{eq:fully-connected-hypergraph}. In this case, the corresponding unitary transformation is  
\begin{align}
    U=\begin{bmatrix}
        0&e^{i\alpha\pi}\\
        e^{-i\alpha\pi}&0
    \end{bmatrix}.
\end{align}
This demonstrates how different choices of $ A $ yield distinct $CU$ interactions.

\begin{figure}
    \centering
    \includegraphics[width=0.6\linewidth]{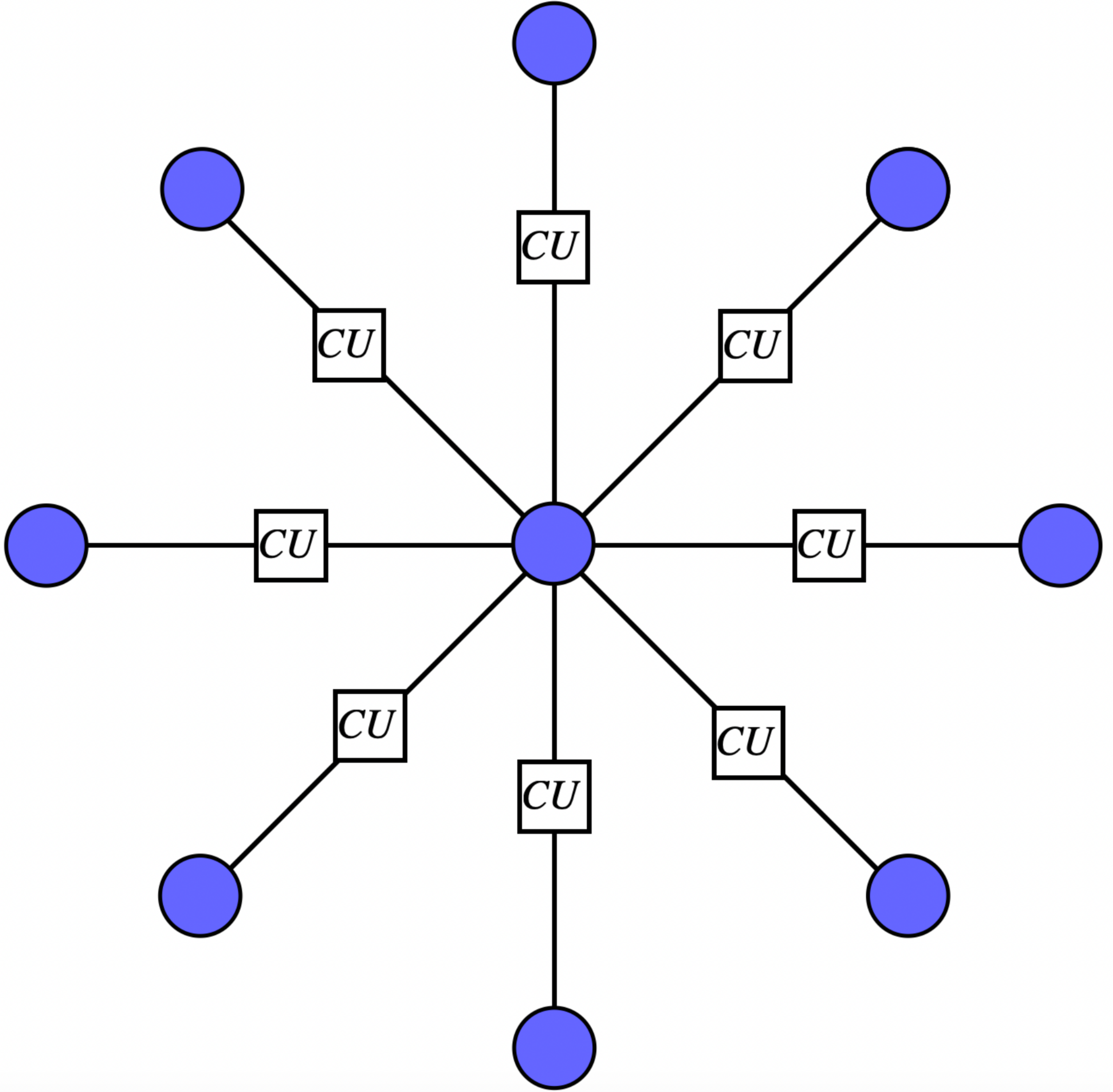}
    \caption{Controlled-Unitary (CU) Star-Shaped Graph. In this representation, the central vertex serves as the target, while all other vertices act as controls. The non-symmetric GHZ state, Eq.~\eqref{eq:CU-GHZ}, is local unitary (LU) equivalent to this CU graph.}
    \label{fig:3}
\end{figure}

\textit{Summary and Outlook.}--- 
We established a graph-based framework for non-symmetric GHZ states, proving their LU equivalence to two structures: fully connected weighted hypergraphs with controlled-phase interactions and star-shaped controlled-unitary ($CU$) graphs. This dual representation bridges weighted hypergraphs and controlled-unitary entanglement structures, offering new insights into non-symmetric multipartite quantum states.

In experiments, decoherence, control errors, and imperfections in high-dimensional quantum systems can lead to non-symmetric states rather than ideal GHZ states. Our framework provides a systematic method to construct stabilizers for such states using a single ancilla, enabling their detection through local operations regardless of system size. Beyond hypergraph-based stabilizers, the CU formalism offers a structured approach to synthesizing non-symmetric quantum states and analyzing controlled quantum operations.

The ability to construct stabilizers for non-symmetric GHZ states may have implications for error correction, as stabilizer based techniques are a cornerstone of fault-tolerant quantum computation. Additionally, the CU framework extends graph-based methods to controlled-unitary interactions, relevant to quantum simulation and networked quantum protocols. Future directions include applying this approach to other resource states, investigating experimental implementations in noisy quantum devices, and exploring potential advantages of non-symmetric entanglement in quantum networks and multipartite cryptography.

\acknowledgements

We thank Mario Flory, Markus Grassl, Otfried G\"{u}hne, Barbara Kraus, G\'{e}za T\'{o}th, and Karol \.{Z}yczkowski for valuable discussions. This work was supported by the Equal Opportunity Program, Grant Line 2: Support for Female Junior Professors and Postdocs through Academic Staff Positions, 14th funding round of Paderborn University.
\vspace{-0.5cm}

\bibliography{GHZ}

\onecolumngrid

\section{Supplementary Material: [Non-symmetric GHZ states: weighted hypergraph and controlled-unitary graph representations]}

Here we provide the proofs to the propositions in the letter, the qudit weighted hypergraph algorithm and an example of a fully connected weighted hypergraph for 5 qubits.

\subsection{Graph Representations of GHZ and Non-Symmetric GHZ States}  

GHZ states are fundamental multipartite entangled states that admit two well-known local unitary (LU) equivalent graph representations: the star-shaped graph and the fully connected graph. These representations differ in topology but encode the same entanglement structure under LU transformations.

\begin{figure}[h]
    \centering
    \includegraphics[width=0.5\linewidth]{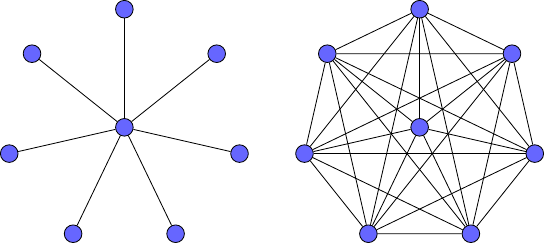}
    \caption{LU-equivalent graph representations of GHZ states. (Left) A star-shaped graph where a central qubit is entangled with all others. (Right) A fully connected graph where each qubit is connected to every other qubit.}
    \label{fig:1}
\end{figure}

Non-symmetric GHZ states, defined as  
\begin{align}
    \ket{n\text{-GHZ}_\alpha}=\cos{\frac{\alpha \pi}{2}}\ket{00\dots 0}+\sin{\frac{\alpha \pi}{2}}\ket{11\dots 1},
\end{align}  
extend this framework by introducing an asymmetry in the superposition amplitudes. Unlike standard GHZ states, these states are not LU equivalent to conventional graph states but can be mapped to fully connected weighted hypergraphs, where multi-qubit controlled-phase ($CZ^\alpha$) interactions encode the entanglement structure.

To illustrate this, we provide the weighted hypergraph representation for the five-qubit non-symmetric GHZ state in Fig.~\ref{fig:5ghz}. The hypergraph decomposition incorporates multi-qubit interactions with weight-dependent phase shifts.

\begin{figure}[h]
    \centering
    \includegraphics[width=0.7\linewidth]{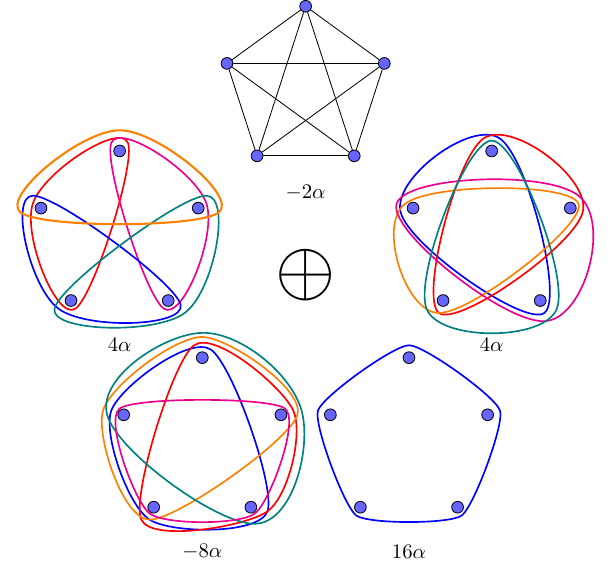}
    \caption{Fully connected weighted hypergraph representation of the five-qubit non-symmetric GHZ state. Each subfigure depicts hyperedges of different orders with corresponding phase weights.}
    \label{fig:5ghz}
\end{figure}

For $n=5$ qubits, the interactions follow:
\begin{itemize}
\item[•] Two-qubit edges ($k=1$): Weight $-2\alpha$  
\item[•] Three-qubit hyperedges ($k=2$): Weight $4\alpha$  
\item[•] Four-qubit hyperedges ($k=3$): Weight $-8\alpha$  
\item[•] Five-qubit hyperedge ($k=4$): Weight $16\alpha$  
\end{itemize}

This provides a graphical way to represent the entanglement structure of non-symmetric GHZ states and highlights the role of weighted multi-qubit interactions in their characterization.

\subsection{Proof of Proposition 1}
\begin{proposition} For a set of vertices $V$, let $E_k$ denote the set of all $k$ element combinations from $V$. The non-symmetric GHZ state $\ket{n\text{-GHZ}_\alpha}$ is LU equivalent to the following fully connected weighted hypergraph state:
\begin{align}\label{eq:fully-connected-hypergraph_supplementary}
    \ket{G}:=\displaystyle\prod_{k=1}^{n-1}\prod_{\textbf{e}\in E_{k+1}}CZ_{\textbf{e}}^{(-2)^{k}\alpha}\ket{+}^{\otimes V},
\end{align}
where each hyperedge $\textbf{e}$ involving $k+1$ vertices carries a weight of $(-2)^{k}\alpha \pi$.
 This equivalence is achieved via the local transformation:  
\begin{align}
    \ket{G}= RZ_1^\dagger(\alpha\pi)P^\dagger(\alpha\pi)^{\otimes V\setminus\{1\}}H^{\otimes V}P^\dagger_1(\pi/2)\ket{n\text{\normalfont-GHZ}_\alpha}\, ,
\end{align}  
where  
\begin{align}
    RZ(\varphi)=\begin{bmatrix}
        e^{-i\varphi/2}&0\\
        0&e^{i\varphi/2}
    \end{bmatrix}, \hspace{0.4cm}
    P(\varphi)=\begin{bmatrix}
        1 &0 \\
        0& e^{i\varphi}
    \end{bmatrix}.
\end{align}  
\end{proposition}
\begin{proof}
    \indent Let us define $\ket{x_1,x_2,\dots x_n}=\ket{ \boldsymbol{x} }$, where $\boldsymbol{x}$ is the decimal representation of $x_1x_2\dots x_n$. Let us
also define $w_{\boldsymbol{x}}=\text{Hamming Weight}(\boldsymbol{x})$, where the Hamming weight is the number of $1$s in $\boldsymbol{x}$. \\
Now we apply the $P^\dagger(\pi/2)$ and $H$ gates.
\begin{align}
    H^{\otimes V}\displaystyle P^\dagger_1 \textstyle(\frac{\pi}{2})\displaystyle\ket{n\text{-GHZ}_\alpha}&=\displaystyle\sum_{\boldsymbol{x} = 0}^{2^n-1}(\cos{\frac{\alpha\pi}{2}}-(-1)^{w_{\boldsymbol{x}}}i\sin{\frac{\alpha\pi}{2}})\ket{ \boldsymbol{x} }
    =    \displaystyle\sum_{\boldsymbol{x} = 0}^{2^n-1}e^{(-1)^{w_{\boldsymbol{x}}+1}i\frac{\alpha\pi}{2}}\ket{ \boldsymbol{x} }
    =\sum_{\boldsymbol{x} = 0}^{2^n-1}\omega^{(-1)^{w_{\boldsymbol{x}}+1}\frac{\alpha}{2}}\ket{ \boldsymbol{x} }.
\end{align}
Next we apply the $RZ^\dagger(\alpha\pi)$ on the first qubit and $P^\dagger(\alpha\pi)$ on the rest.
\begin{align}
    \ket{\psi}&=P^\dagger(\alpha\pi)^{\otimes V\setminus\{1\}}RZ_1^\dagger(\alpha\pi)\displaystyle\sum_{\boldsymbol{x} = 0}^{2^n-1}\omega^{(-1)^{w_{\boldsymbol{x}}+1}\frac{\alpha}{2}}\ket{ \boldsymbol{x} }=\displaystyle\sum_{\boldsymbol{x} = 0}^{2^n-1} \omega^{\alpha\left(\frac{(-1)^{w_{\boldsymbol{x}}+1}}{2}-w_{\boldsymbol{x}}+\frac{1}{2}\right)}\ket{ \boldsymbol{x} }.
\end{align}
Now we consider the fully connected hypergraph state,
\begin{align}
    \ket{G}&=\displaystyle\prod_{k=1}^{n-1}\prod_{\textbf{e}\in E_{k+1}}CZ_{\textbf{e}}^{(-2)^{k}\alpha}\ket{+}^{\otimes V}=\displaystyle\sum_{\boldsymbol{x} = 0}^{2^n-1}\prod_{k=1}^{n-1}\prod_{\textbf{e}\in E_{k+1}}CZ_{\textbf{e}}^{(-2)^{k}\alpha}\ket{ \boldsymbol{x} }.
\end{align}
For a given $\ket{ \boldsymbol{x} }$ and $k$, upon the application of the $CZ$s for all $\textbf{e}\in E_{k+1}$ we get
\begin{align}
    \ket{ \boldsymbol{x} }\rightarrow \omega^{\beta_k}\ket{ \boldsymbol{x} },
\end{align}
where $\beta_k=(-2)^{k}\alpha\displaystyle\sum_{\textbf{e}\in E_{k+1}}\prod_{j \in \textbf{e}}x_j$.\\
As only $w_{\boldsymbol{x}}$ of the $x_j$'s are $1$ there are only ${w_{\boldsymbol{x}}\choose k+1}$ non zero terms in the sum and the non zero ones all have a value of 1. Therefore, we can simplify the exponent to,
\begin{align}
    \beta_k=(-2)^{k}{w_{\boldsymbol{x}}\choose k+1}\alpha.
\end{align}

\noindent We now can perform the product over all $k$ , which will give us the following sum in the exponent of $\omega$,
\begin{align}
    \sum_{k=1}^{n-1}\beta_{k}=\alpha\sum_{k=1}^{n-1}(-2)^{k}{w_{\boldsymbol{x}}\choose k+1}.
\end{align}
We can simplify this further by noticing that ${w_{\boldsymbol{x}}\choose k}=0$ if $k>w_{\boldsymbol{x}}$. Hence,
\begin{align}
    \alpha\sum_{k=1}^{n-1}(-2)^{k}{w_{\boldsymbol{x}}\choose k+1}&=-\frac{\alpha}{2}\sum_{k=2}^{w_{\boldsymbol{x}}}(-2)^{k}{w_{\boldsymbol{x}}\choose k}=-\frac{\alpha}{2}\left(\sum_{k=0}^{w_{\boldsymbol{x}}}(-2)^{k}{w_{\boldsymbol{x}}\choose k}-1+2w_{\boldsymbol{x}}\right)\nonumber\\
    &=-\frac{\alpha}{2}((1-2)^{w_{\boldsymbol{x}}}-1+2w_{\boldsymbol{x}})=\alpha\left(\frac{(-1)^{w_{\boldsymbol{x}}+1}}{2}-w_{\boldsymbol{x}}+\frac{1}{2}\right),
\end{align}
where we have used the binomial theorem to obtain the 3\textsuperscript{rd} equality. \\
\par
\noindent Writing the graph state using the above, we find it equal to $\ket{\psi}$,
\begin{align}
    \ket{G}=\displaystyle\sum_{\boldsymbol{x} = 0}^{2^n-1}\omega^{\alpha\left(\frac{(-1)^{w_{\boldsymbol{x}}+1}}{2}-w_{\boldsymbol{x}}+\frac{1}{2}\right)}\ket{ \boldsymbol{x} }=\ket{\psi}.
\end{align}
Therefore,
\begin{align}
    \ket{G}= RZ_1^\dagger(\alpha\pi)P^\dagger(\alpha\pi)^{\otimes V\setminus\{1\}}H^{\otimes V}P^\dagger_1(\pi/2)\ket{n\text{\normalfont-GHZ}_\alpha}.
\end{align}
\end{proof}

\subsection{Proof of Proposition 2}
\begin{proposition}
    Given the weighted hypergraph state of $n=|V|$ qubits,
    \begin{align}
    \ket{G}=\displaystyle\prod_{k=1}^{n-1}\prod_{\textbf{e}\in E_{k+1}}CZ_{\textbf{e}}^{(-2)^{k}\alpha}\ket{+}^{\otimes V},
    \end{align}
    and a single ancilla, it is possible to stabilize the $n+1$ qubit state, if the ancilla is connected to each of the $n$ qubits through a $CZ$ operation.
    The stabilizers are given by,
    \begin{align}
    K_{G^\prime}^{(1)}&=X\otimes Z\otimes Z \otimes... \otimes Z\nonumber\\
    K_{G^\prime}^{(2)}&=X^\alpha Z X^{-\alpha}\otimes P^\dagger(\alpha\pi)XP(\alpha\pi)\otimes I \otimes... \otimes I\nonumber\\
    K_{G^\prime}^{(3)}&=X^\alpha Z X^{-\alpha}\otimes I\otimes P^\dagger(\alpha\pi)XP(\alpha\pi) \otimes... \otimes I\\
    &\hspace{0.2cm}\vdots\nonumber\\
    K_{G^\prime}^{(n+1)}&=X^\alpha Z X^{-\alpha}\otimes I\otimes I \otimes... \otimes P^\dagger(\alpha\pi)XP(\alpha\pi),\nonumber
\end{align}
where 
\begin{align}
    X^\alpha:=\frac{1}{2}((1-e^{i\pi\alpha})I+(1+e^{i\pi\alpha})X),
\end{align}
\label{eq:X-alpha}
the first qubit is the ancilla and $G^\prime$ is the graph with the ancilla.
\end{proposition}
\begin{proof}
    We employ a technique introduced by Tsimakuridze and G\"{u}hne \cite{Tsimakuridze_2017}, which establishes LU equivalence between graph states and specific weighted hypergraph states. This technique allows us to introduce a single ancilla qubit, enabling a stabilizer representation for the fully connected weighted hypergraph state. 

Consider a hypergraph $G^\prime$ and define the real power $\alpha$ of the Pauli-X gate as:
\begin{align}
    X^\alpha:=\frac{1}{2}((1-e^{i\pi\alpha})I+(1+e^{i\pi\alpha})X).
\end{align}
The application of $X^\alpha$ upon any vertex $v_i$ of $G^\prime$ leads to a weighted hypergraph given by the following rules,
\begin{enumerate}
    \item For each $\textbf{e}$ in $A_i=\{\textbf{e}|\textbf{e}/{v_i}\in E\}$ add a weight $\alpha$ to $\textbf{e}$.
    \item For each $k\leq n$, $\textbf{e}_1,\textbf{e}_2,..\textbf{e}_k\in A_i$, $\textbf{e}_1\neq \textbf{e}_2...\neq \textbf{e}_k$ add weight $(-2)^k\alpha$ to the edge $\textbf{e}_1\cup \textbf{e}_2...\cup \textbf{e}_k$.
\end{enumerate}

We now introduce a star-shaped graph with $n+1$ qubits, where the ancilla qubit (labeled 1) is connected to all $n$ system qubits via CZ operations:
\begin{align}
    E=\{\{1,2\},\{1,3\},...,\{1,n+1\}\}.
\end{align}
Now we consider the application of $X^\alpha$ on the ancilla. In accordance with the two rules given above we obtain,
\begin{align}
    A_1=\{\{2\},\{3\},...,\{n+1\}\}.
\end{align}
By Rule 1, all qubits except the ancilla receive local phase gates. The second rule tells us that for each $k$ combinations of the $n$ qubits in $A_1$, we add a hyperedge of weight $(-2)^k\alpha$ between those $k$ particles. However, if we ignore the ancilla qubit this is exactly the construction of the weighted hypergraph in Eq.~\eqref{eq:fully-connected-hypergraph_supplementary}.

To construct the stabilizers, we use the graph state stabilizer formalism:
\begin{align}
    K_G^{(j)}=X_{j} \prod_{k\in V}Z_k^{\Gamma_{jk}}.
\end{align}
Now we simply conjugate these with the operator,
\begin{align}
    X^\alpha\otimes P^\dagger(\alpha\pi)\otimes \dots \otimes P^\dagger(\alpha\pi).
\end{align}
where $X^\alpha$ is due to the original application and the phase gates are due to rule 1.
This leads to a new set of local stabilizers:
\begin{align}
    K_{G^\prime}^{(1)}&=X\otimes Z\otimes Z ... \otimes Z\nonumber\\
    K_{G^\prime}^{(2)}&=X^\alpha Z X^{-\alpha}\otimes P^\dagger(\alpha\pi)XP(\alpha\pi)\otimes I ... \otimes I\nonumber\\
    K_{G^\prime}^{(3)}&=X^\alpha Z X^{-\alpha}\otimes I\otimes P^\dagger(\alpha\pi)XP(\alpha\pi) ... \otimes I\\
    ...&\nonumber\\
    K_{G^\prime}^{(n+1)}&=X^\alpha Z X^{-\alpha}\otimes I\otimes I ... \otimes P^\dagger(\alpha\pi)XP(\alpha\pi).\nonumber
\end{align}

\end{proof}

This construction demonstrates that a single ancilla qubit, connected to all system qubits via controlled-Z operations, is sufficient to stabilize the fully connected weighted hypergraph state. The local nature of these stabilizers allows for direct experimental verification and manipulation, overcoming the primary challenge of multi-qubit entanglement in weighted hypergraphs. Thus, we have established a novel method for stabilizing non-symmetric GHZ states using a fixed number of ancillae—independent of system size $n$—which is crucial for practical quantum information processing.

\subsection{Proof of Proposition 3}
\begin{proposition}
    Given a qudit $X^\alpha$,
    \begin{align}
    X^\alpha&=H^\dagger Z^\alpha H=\displaystyle\sum_{k,l=0}^{d-1}\omega^{l(\alpha-k)}X^k,
\end{align}
we can define a qudit non-symmetric GHZ,
\begin{align}
    \ket{n\text{-GHZ}^d_\alpha}&:=\displaystyle\prod_{k=2}^{n}CX_{1k}X^\alpha_1\ket{0}^{\otimes n} =\sum_{k,l=0}^{d-1}\omega^{l(\alpha-k)}\ket{kk....k},
\end{align}
which is LU equivalent to a qudit weighted hypergraph given by the operation,
\begin{align}
    CZ_{G}^{-\alpha}=(Z_1Z_2...\,Z_{n})^{-\alpha}.
\end{align}
\end{proposition}
\begin{proof}
    We begin by considering the general multi-controlled-Z operation acting on a hyperedge $ \mathbf{e} $ in a qudit system:
\begin{align}
    CZ_\mathbf{e}=\sum_{k_{e_1},...,k_{e_m}=0}^{d-1} \omega^{\bar{k}_\mathbf{e}}\ket{k_{e_1}...k_{e_m}}\bra{k_{e_1}...k_{e_m}},
\end{align}
where the total phase contribution for a given basis state is defined as  
\begin{align}
    \bar{k}_\mathbf{e} = \prod_{j=1}^{m} k_{e_j}.
\end{align}
A key observation is how this operator commutes with the generalized $X$ gate at some vertex $e_1$. We have
\begin{align}
    X_{e_1} CZ_\mathbf{e}X_{e_1}^\dagger = \sum_{k_{e_1},...k_{e_m}=0}^{d-1} \omega^{\bar{k}_\mathbf{e}} \ket{k_{e_1}+1,...,k_{e_m}}\bra{k_{e_1}+1,...,k_{e_m}},
\end{align}
and by multiplying on the left by $ CZ_\mathbf{e}^\dagger $, we obtain

\begin{align}
    CZ_\textbf{e}^\dagger X_{e_1}CZ_\textbf{e}X_{e_1}^\dagger&=\displaystyle\sum_{l_{e_1},...l_{e_m}=0}^{d-1}\sum_{k_{e_1},...k_{e_m}=0}^{d-1}\omega^{\bar{k}_\textbf{e}-\bar{l}_\textbf{e}}\ket{l_{e_1},...l_{e_m}}\braket{l_{e_1},...l_{e_m}|k_{e_1}+1,...k_{e_m}}\bra{k_{e_1}+1,...k_{e_m}}\nonumber\\
    &=\displaystyle\sum_{l_{e_1},...l_{e_m}=0}^{d-1}\sum_{k_{e_1},...k_{e_m}=0}^{d-1}\omega^{\bar{k}_\textbf{e}-\bar{l}_\textbf{e}}\delta_{l_{e_1}-1,k_{e_1}}\prod_{j=2}^m \delta_{l_{e_j},k_{e_j}}\ket{l_{e_1},...l_{e_m}}\bra{k_{e_1}+1,...k_{e_m}}\\
    &=\displaystyle\sum_{l_{e_1}=1}^{d-1}\sum_{l_{e_2},...l_{e_m}=0}^{d-1}\omega^{-\bar{l}_{\textbf{e}/\{e_1\}}}\ket{l_{e_1}...l_{e_m}}\bra{l_{e_1}...l_{e_m}}+\sum_{l_{e_2},...l_{e_m}=0}^{d-1}\omega^{(d-1)\bar{l}_{\textbf{e}/\{e_1\}}}\ket{0...l_{e_m}}\bra{0...l_{e_m}}\nonumber\\
    &=\displaystyle\sum_{l_{e_1}=0}^{d-1}\ket{l_{e_1}}\bra{l_{e_1}}\otimes\sum_{l_{e_2},...l_{e_m}=0}^{d-1}\omega^{-\bar{l}_{\textbf{e}/\{e_1\}}}\ket{l_{e_2}...l_{e_m}}\bra{l_{e_2}...l_{e_m}}=\mathbb{I}_{e_1}\otimes CZ_{\textbf{e}/\{e_1\}}\nonumber , 
\end{align}

where to get the last line we have used $\omega^d=1$. Therefore, we have,
\begin{align}
    X_{e_1}CZ_\textbf{e}=CZ_\textbf{e}(\mathbb{I}_{e_1}\otimes CZ_{\textbf{e}/\{e_1\}})X_{e_1} \, .
\end{align}

Next, we derive the explicit form of $ X^\alpha $ using the following qudit operators and relations,
\begin{equation}
    \begin{aligned}
        H &= \sum_{j,k=0}^{d-1} \omega^{jk} \ket{j} \bra{k}\qquad\qquad &
        Z &= \sum_{j=0}^{d-1} \omega^j \ket{j} \bra{j} \\
        X &= \sum_{j=0}^{d-1} \ket{j+1} \bra{j}\qquad\qquad &
        X &= H^\dagger ZH \, .
    \end{aligned}
\end{equation}
Using the above equations we can write,
\begin{align}
    X^\alpha&=H^\dagger Z^\alpha H=\displaystyle\sum_{j^\prime,k^\prime,j,k,l=0}^{d-1} \omega^{jk}\omega^{-j^\prime k^\prime}\omega^{\alpha l}\ket{j^\prime}\braket{k^\prime|l}\braket{l|j}\bra{k}=\displaystyle\sum_{j^\prime,k^\prime,j,k,l=0}^{d-1} \omega^{jk}\omega^{-j^\prime k^\prime}\omega^{\alpha l} \delta_{k^\prime,l}\delta_{j,l}\ket{j^\prime}\bra{k}\nonumber\\
    &=\displaystyle\sum_{j^\prime,k,l=0}^{d-1} \omega^{lk}\omega^{-j^\prime l}\omega^{\alpha l}\ket{j^\prime}\bra{k}=\displaystyle\sum_{j,k,l=0}^{d-1} \omega^{l(j-(k+j)+\alpha)}\ket{j+k}\bra{j}=\displaystyle\sum_{j,k,l=0}^{d-1} \omega^{l(\alpha-k)}\ket{j+k}\bra{j}=\displaystyle\sum_{k,l=0}^{d-1}\omega^{l(\alpha-k)}X^k .
\end{align}
Similar to the qubit case we consider a qudit hypergraph $G(V,E)$ with
\begin{align}
    \ket{G}=\displaystyle\prod_{\textbf{e}\in E} CZ_{\textbf{e}}\ket{+}^{\otimes n}=CZ_G\ket{+}^{\otimes n}.
\end{align}
Additionally, we define the graph $\Delta_i G(V,\Delta_i E)$, where
\begin{align}
    \Delta_i E=\{\textbf{e}/{i}|i\in \textbf{e}, \textbf{e}\in E\}.
\end{align}
Using the commutation relation,
\begin{align}
    X^\alpha_i \ket{G}&=\displaystyle\sum_{k,l}\omega^{l(\alpha-k)}X_i^k CZ_G\ket{+}^{\otimes n}=CZ_G\sum_{k,l}\omega^{l(\alpha-k)}CZ_{\Delta_i G}^{\dagger k} X_i^k\ket{+}^{\otimes n}\\
    &=\sum_{k,l}\omega^{l(\alpha-k)}CZ_{\Delta_i G}^{\dagger k} CZ_G\ket{+}^{\otimes n}=CZ_{\Delta_i G}^{-\alpha}\ket{G}\nonumber
\end{align}
\noindent
Therefore, like in the qubit case the action of $X^\alpha$ has introduced new weighted hyperedges to our hypergraph.\\
Similar to the qubit case we consider the qudit \emph{star-shape graph} $G(V,E)$,
\begin{align}
    E=\{\{1,2\},\{1,3\},...,\{1,n+1\}\}.
\end{align}
We apply $X^\alpha$ on the ancilla qudit, i.e. qudit 1, which gives us,
\begin{align}
    \Delta_1 E=\{\{2\},\{3\},...,\{n+1\}\}.
\end{align}
Therefore, the additional weighted hypergraph added is,
\begin{align}\label{eq:qudit hypergraph}
    CZ_{\Delta_1 G}^{-\alpha}=(Z_2Z_3...\,Z_{n+1})^{-\alpha},
\end{align}
as an edge with a single vertex is simply a local operation. This does not look like a weighted hypergraph as $Z$ is a local operation,
however, in general (see next section),
\begin{align}
    (Z_2Z_3...Z_{n+1})^{-\alpha}\neq Z_2^{-\alpha}Z_3^{-\alpha}...\,Z_{n+1}^{-\alpha}.
\end{align}
To expand the brackets we need to include non-local operations. For example in the qubit case those non-local operations are what give rise to the fully connected weighted hypergraph. However, in the case of qudits, we find that
the non-local operations are not of the form of $CZ_{\textbf{e}}^\alpha$ but are of the form,
\begin{align}
    CP_e(\vec{\alpha})=\displaystyle\sum_{k_{e_1},...k_{e_m}=0}^{d-1}\omega^{\vec{\alpha}_{k_e}}\ket{k_{e_1}...k_{e_m}}\bra{k_{e_1}...k_{e_m}},
\end{align}
where $k_e$ is the decimal representation of $k_{e_1}k_{e_2}...k_{e_m}$. Therefore, $CP_e(\vec{\alpha})$ is a diagonal matrix of phases decided by the vector $\vec{\alpha}$. These don't follow a pattern like the qubit case, and therefore we need to use an algorithm to find the exact operations
to apply (see next section).\\
For our choice of the specific form of the non-symmetric qudit GHZ, we note that for qubits,
\begin{align}
    \ket{n\text{-GHZ}_\alpha}=\displaystyle \omega^{\alpha/2}\prod_{k=2}^{n}CX_{1k}X_1X^\alpha_1\ket{0}^{\otimes n}.
\end{align}
Therefore, we define the non-symmetric qudit GHZ state, $\ket{n\text{-GHZ}^d_\alpha}$ as,
\begin{align}
    \ket{n\text{-GHZ}^d_\alpha}&:=\displaystyle\prod_{k=2}^{n}CX_{1k}X^\alpha_1\ket{0}^{\otimes n}=\sum_{k,l=0}^{d-1}\omega^{l(\alpha-k)}\ket{kk....k}.
\end{align}

Now let us show that this state is indeed local unitary equivalent to the weighted hypergraph from Eq.~\eqref{eq:qudit hypergraph}. We apply Hadamards on all the qudits,
\begin{align}
    H^{\otimes n}\ket{n\text{-GHZ}_\alpha^d}=\sum_{k,l=0}^{d-1}\omega^{l(\alpha-k)}H^{\otimes n}\ket{kk....k}=\sum_{k,l=0}^{d-1}\omega^{l(\alpha-k)}\sum_{\boldsymbol{x} = 0}^{2^d-1}\omega^{kw_{\boldsymbol{x}}}\ket{ \boldsymbol{x} }.
\end{align}

\noindent Writing $w_{\boldsymbol{x}}=dq_{\boldsymbol{x}}+r_{\boldsymbol{x}}$ for some integers $q_{\boldsymbol{x}}$ and $0\leq r_{\boldsymbol{x}} \leq d-1$ the state can be written as,
\begin{align}
    \sum_{\boldsymbol{x} = 0}^{2^d-1}\sum_{k,l=0}^{d-1}\omega^{l\alpha}\omega^{k(r_{\boldsymbol{x}}-l)}\ket{ \boldsymbol{x} }= \sum_{\boldsymbol{x} = 0}^{2^d-1}\sum_{l=0}^{d-1}\omega^{l\alpha}\delta_{r_{\boldsymbol{x}},l}\ket{ \boldsymbol{x} }=\sum_{\boldsymbol{x} = 0}^{2^d-1}\omega^{r_{\boldsymbol{x}}\alpha}\ket{ \boldsymbol{x} }.
\end{align}
We can show that this is exactly the state for the weighted hypergraph $\Delta_1 G$. We consider the following
\begin{align}
    Z_1Z_2...\,Z_{n}\ket{+}^{\otimes n}=\sum_{x_1,x_2,...,x_n=0}^{d-1}\bigotimes_{j=1}^n\omega^{x_j}\ket{x_j}=\sum_{\boldsymbol{x} = 0}^{2^d-1}\omega^{w_{\boldsymbol{x}}}\ket{ \boldsymbol{x} }=\sum_{\boldsymbol{x} = 0}^{2^d-1}\omega^{r_{\boldsymbol{x}}}\ket{ \boldsymbol{x} }.
\end{align}
Therefore, as $ Z_1Z_2...\,Z_{n}$ is diagonal we get, 
\begin{align}
    (Z_1Z_2...\,Z_{n})^{\alpha}=\sum_{\boldsymbol{x} = 0}^{2^d-1}\omega^{\alpha r_{\boldsymbol{x}}}\ket{ \boldsymbol{x} }.
\end{align}
And hence the weighted hypergraph obtained from applying $X^{-\alpha}$ on the central qudit of a \emph{star-shaped graph} and then tracing out the ancilla qudit, is equivalent to $\ket{n\text{-GHZ}_\alpha^d}$.\\
\end{proof}
\par

\noindent The stabilizers are given similarly to the qubit case,
\begin{align}
    K_{G^\prime}^{(1)}&=X\otimes Z\otimes Z ... \otimes Z\nonumber\\
    K_{G^\prime}^{(2)}&=X^{-\alpha} Z X^{\alpha}\otimes P^\dagger(\alpha\pi)XP(\alpha\pi)\otimes I ... \otimes I\nonumber\\
    K_{G^\prime}^{(3)}&=X^{-\alpha} Z X^{\alpha}\otimes I\otimes P^\dagger(\alpha\pi)XP(\alpha\pi) ... \otimes I\\
    ...&\nonumber\\
    K_{G^\prime}^{(n+1)}&=X^{-\alpha} Z X^{\alpha}\otimes I\otimes I ... \otimes P^\dagger(\alpha\pi)XP(\alpha\pi).\nonumber
\end{align}

The stabilizer formalism for non-symmetric GHZ states in qudit systems is highly analogous to that in qubit systems. The key difference between the qudit and qubit cases is the sign of the parameter $\alpha$ in the power of $X$. This distinction arises due to the Hermitian conjugate in the commutation relation between the generalized $X$ operator and the $CZ_{\textbf{e}}$.

\subsection{Qudit weighted hypergraph algorithm}

Upon the application of $X^\alpha$ on a qudit hypergraph, the resultant weighted hypergraph requires general $CP$ edges to be added. To find the $CP$ operations we do the following in the specific case of the star-shape graph as the starting graph. The weighted edges added as shown in the previous section
are given by
\begin{align}
    C_{\Delta_iG}^\alpha=(Z_1Z_3...Z_n)^\alpha=Z_1^\alpha Z_3^\alpha...Z_n^\alpha...\text{(higher order terms)}.
\end{align}
As an example for why the higher order terms are needed, consider the following example for qutrits. 
    \begin{align}
        Z_{1}=Z_{2}=\begin{pmatrix}
            1&&\\
            &\omega&\\
            &&\omega^2
        \end{pmatrix}
    \end{align}
    Then,
    \begin{align}
       Z_1Z_2&=\setlength{\arraycolsep}{3pt}\begin{pmatrix}
           1\text{ }&&&&&&&&\\
           &\omega&&&&&&&\\
           &&\omega^2&&&&&&\\
           &&&1\text{ }&&&&&\\
           &&&&\omega&&&&\\
           &&&&&\omega^2&&&\\
           &&&&&&1\text{ }&&\\
           &&&&&&&\omega&\\
           &&&&&&&&\omega^2
       \end{pmatrix}\setlength{\arraycolsep}{3pt}\begin{pmatrix}
           1\text{ }&&&&&&&&\\
           &1\text{ }&&&&&&&\\
           &&1\text{ }&&&&&&\\
           &&&\omega&&&&&\\
           &&&&\omega&&&&\\
           &&&&&\omega&&&\\
           &&&&&&\omega^2&&\\
           &&&&&&&\omega^2&\\
           &&&&&&&&\omega^2
       \end{pmatrix}=
       \setlength{\arraycolsep}{3pt}\begin{pmatrix}
           1\text{ }&&&&&&&&\\
           &\omega&&&&&&&\\
           &&\omega^2&&&&&&\\
           &&&\omega&&&&&\\
           &&&&\omega^2&&&&\\
           &&&&&1\text{ }&&&\\
           &&&&&&\omega^2&&\\
           &&&&&&&1\text{ }&\\
           &&&&&&&&\omega
       \end{pmatrix}
    \end{align}
    where we have used that $\omega^3=1$. Therefore,
    \begin{align}
        (Z_1Z_2)^\alpha=\setlength{\arraycolsep}{1pt}\begin{pmatrix}
           1\text{ }&&&&&&&&\\
           &\omega^\alpha&&&&&&&\\
           &&\omega^{2\alpha}&&&&&&\\
           &&&\omega^\alpha&&&&&\\
           &&&&\omega^{2\alpha}&&&&\\
           &&&&&1\text{ }&&&\\
           &&&&&&\omega^{2\alpha}&&\\
           &&&&&&&1\text{ }&\\
           &&&&&&&&\omega^\alpha
       \end{pmatrix}
     \end{align}
     however,
     \begin{align}
       Z_1^\alpha Z_2^\alpha&=\setlength{\arraycolsep}{1pt}\begin{pmatrix}
           1\text{ }&&&&&&&&\\
           &\omega^\alpha&&&&&&&\\
           &&\omega^{2\alpha}&&&&&&\\
           &&&1\text{ }&&&&&\\
           &&&&\omega^\alpha&&&&\\
           &&&&&\omega^{2\alpha}&&&\\
           &&&&&&1\text{ }&&\\
           &&&&&&&\omega^\alpha&\\
           &&&&&&&&\omega^{2\alpha}
       \end{pmatrix}\setlength{\arraycolsep}{1pt}\begin{pmatrix}
           1\text{ }&&&&&&&&\\
           &1\text{ }&&&&&&&\\
           &&1\text{ }&&&&&&\\
           &&&\omega^\alpha&&&&&\\
           &&&&\omega^\alpha&&&&\\
           &&&&&\omega^\alpha&&&\\
           &&&&&&\omega^{2\alpha}&&\\
           &&&&&&&\omega^{2\alpha}&\\
           &&&&&&&&\omega^{2\alpha}
       \end{pmatrix}=
       \setlength{\arraycolsep}{1pt}
        \begin{pmatrix}
           1\text{ }&&&&&&&&\\
           &\omega^\alpha&&&&&&&\\
           &&\omega^{2\alpha}&&&&&&\\
           &&&\omega^\alpha&&&&&\\
           &&&&\omega^{2\alpha}&&&&\\
           &&&&&\omega^{3\alpha}&&&\\
           &&&&&&\omega^{2\alpha}&&\\
           &&&&&&&\omega^{3\alpha}&\\
           &&&&&&&&\omega^{4\alpha}
       \end{pmatrix}
     \end{align}
    Hence we require the higher order correcting terms.
The higher order terms need to be calculated using an algorithm, which we provide here. We would like to be able to split the higher order correcting terms based on the number of qudits that they act on. Hence, we do the following,
\begin{enumerate}
    \item Consider a $d^2\times d^2$ identity matrix $I_{ij}$ acting on two particles $i$ and $j$. For the state $\ket{x_ix_j}$, calculate the number $k$ of $\omega^d$s in the phase of the state upon the application of $Z_1^\alpha Z_2^\alpha...Z_n^\alpha$. Add a phase $\omega^{-kd\alpha}$ to the $\ket{x_ix_j}\bra{x_ix_j}$ element of $I_{ij}$.
    \item Repeat for all possible values of $x_i$ and $x_j$ such that $x_i,x_j\neq 0$. The obtained matrices $I_{ij}$ are the correction that needs to be applied between all possible two qudit combinations $i$ and $j$.
    \item Consider a $d^3\times d^3$ identity matrix $I_{ijl}$ acting on three particles $i$, $j$, $l$. For the state $\ket{x_ix_jx_l}$,
    calculate the number $k$ of $\omega^d$s in the phase of the state upon the application of $Z_1^\alpha Z_2^\alpha...Z_n^\alpha$. Add a phase $\omega^{-kd\alpha}$ to the $\ket{x_ix_jx_l}\bra{x_ix_jx_l}$ element of $I_{ijl}$.
    \item For the state $\ket{x_ix_jx_l}$, consider the phase $\omega^p$ upon the application of $I_{ij}I_{il}I_{jl}$, i.e. all possible combinations of 2 from $i,j,l$. Add a phase $\omega^{-p}$ to the to the $\ket{x_ix_jx_l}\bra{x_ix_jx_l}$ element of $I_{ijl}$.
    \item Repeat for all possible values of $x_i$, $x_j$ and $x_l$ such that $x_i,x_j,x_l\neq 0$. The obtained matrices $I_{ijl}$ are the correction that needs to be applied between all possible three qudit combinations $i$ $j$ and $l$.
    \item Repeat the above recursively until the number of qudits being applied on is $n$.

\end{enumerate}
The resultant matrices $I_{i_1i_2...i_r}$ for $2\leq r\leq n$ are the higher order correcting terms which give rise to the weighted hypergraph upon the application of $X^\alpha$ on the qudit star-shape graph.

\subsection{Proof of Proposition 4}
\begin{proposition}
        Given the general non-symmetric qudit GHZ,
    \begin{align}
    \ket{n\text{-GHZ}_{\vec{a}}}=\displaystyle\sum_{j=0}^{d-1}a_j\ket{jj\dots j},
    \end{align}
    there exists a LU transformation that takes $\ket{n\text{-GHZ}_{\vec{a}}}$ to a star-shaped controlled-unitary-$U$ graph, for
    \begin{align}
    U=HA^\dagger Z A H^\dagger,
\end{align}
where the operator $A$ is defined as,
\begin{align}
    A\ket{0}=\displaystyle\sum_{j=0}^{d-1}a_j\ket{j}.
\end{align}
\end{proposition}
\begin{proof}
   First let us consider the construction of $\ket{n\text{-GHZ}_{\vec{a}}}$ using the $A$ operator and controlled-$X$ ($CX$) operations $n-1$ times,
    \begin{align}
    \ket{n\text{-GHZ}_{\vec{a}}}=\displaystyle\sum_{j=0}^{d-1}a_j\ket{jj\dots j}=\displaystyle\prod_{k=2}^n CX_{1k}A_1\ket{00\dots 0}.
    \end{align}
    Now let us apply Hadamards on all the qudits except the first one.
\begin{align}
    \displaystyle\prod_{l=2}^n H_l \ket{n\text{-GHZ}_{ns}}=\sum_{j=0}^{d-1}a_j\ket{j}\sum_{k=0}^{d^{n-1}-1}\omega^{j(k_1+\dots+k_{n-1})}\ket{k} =\sum_{j=0}^{d-1}a_j\ket{j}\sum_{k=0}^{d^{n-1}-1}\omega^{jw_k}\ket{k}.
\end{align}
Next we apply $CZ^\dagger$ operators $n-1$ times with the control on the first qudit and the targets as the rest.
\begin{align}
    \displaystyle \prod_{l=2}^{n}CZ_{1l}^\dagger \sum_{j=0}^{d-1}a_j\ket{j}\sum_{k=0}^{d^{n-1}-1}\omega^{jw_k}\ket{k}=\sum_{j=0}^{d-1}a_j\ket{j}\sum_{k=0}^{d^{n-1}-1}\ket{k}.
\end{align}
Lastly we apply $HA^\dagger$ to the first qudit,
\begin{align}
    H_1A_1^\dagger \sum_{j=0}^{d-1}a_j\ket{j}\sum_{k=0}^{d^{n-1}-1}\ket{k}=\sum_{j=0}^{d-1}\ket{j}\sum_{k=0}^{d^{n-1}-1}\ket{k}=\sum_{k^\prime=0}^{d^{n}-1}\ket{k^\prime}=\ket{+}^{\otimes n},
\end{align}
where $k^\prime$ is the decimal representation of $jk_1k_2...\,k_{n-1}$.
Hence, we have found that the operation,
\begin{align}
    H_1A_1^\dagger\displaystyle\prod_{l=2}^n CZ_{1l}^\dagger H_l.
\end{align}
Takes $\ket{n\text{-GHZ}_{\vec{a}}}$ to $\ket{+}^{\otimes n}$. We can rewrite this operation as,
\begin{align}
    H_1A_1^\dagger\displaystyle\prod_{l=2}^n CZ_{l1}^\dagger H_l A_1 H_1^\dagger H_1A_1^\dagger=\prod_{l=2}^n CU_{l1}^\dagger H_l H_1 A_1^\dagger,
\end{align}
where
\begin{align}
    U=HA^\dagger Z A H^\dagger.
\end{align}
Therefore, a $CU$ \emph{star-shaped graph state} is LU equivalent to $\ket{n\text{-GHZ}_{\vec{a}}}$.
\end{proof}

\end{document}